\documentclass[reqno,tbtags,intlimits,a4paper,oneside,12pt]{amsart}
\usepackage[cp1251]{inputenc}%
\usepackage[english]{babel}
\usepackage{amssymb,upref,calc,url}
\usepackage[mathscr]{eucal}
\usepackage{graphicx}
\usepackage{amsmath,amscd,amsthm,verbatim}
\usepackage[all]{xy}
\usepackage[in]{fullpage}
\newtheorem{thm}{\hskip\parindent Theorem}[section]
\newtheorem{lem}[thm]{\hskip\parindent Lemma}
\newtheorem{rem}[thm]{\hskip\parindent Remark}
\newtheorem{prb}[thm]{\hskip\parindent Problem}
\newtheorem{cor}[thm]{\hskip\parindent Corollary}

\theoremstyle{definition}
\newtheorem{dfn}{\hskip\parindent Definition}[section]

\DeclareMathOperator{\wt}{wt}

\newcommand{\CU}{\mathcal{C}(z,\lambda)}

\begin{document}

\title{Sigma functions and Lie algebras
of~Schr\"odinger~operators}
\author{V.\,M.~Buchstaber, E.\,Yu.~Bunkova}
\address{Steklov Mathematical Institute of Russian Academy of Sciences, Moscow, Russia}
\email{buchstab@mi-ras.ru, bunkova@mi-ras.ru}
\keywords{Schr\"odinger operators, polynomial Lie algebras, differentiation of abelian functions over parameters}
\thanks{This work is supported by the Russian Science Foundation under grant 20-11-19998.
}

\begin{abstract}
In the work by V.\,M.\,Buchstaber and  D.\,V.\,Leikin \cite{Nonhol} for any $g > 0$ is defined a system of $2g$ multidimensional Schr\"odinger equations in magnetic fields with quadratic potentials. This systems are equivalent to systems of heat equations in nonholonomic frame. It is proved that such a system determines the sigma function of the universal hyperelliptic curve of genus $g$. A polynomial Lie algebra with $2g$ Schr\"odinger operators $Q_0, Q_2, \ldots, Q_{4g-2}$ as generators is introduced.

In this work for any $g > 0$ we obtain explicit expressions for $Q_0$, $Q_2$, $Q_4$, and recurrent formulas for $Q_{2k}$ with $k>2$ expressing this operators as elements of the polynomial Lie algebra using Lie brackets of the operators $Q_0$, $Q_2$, and $Q_4$.

As an application we obtain explicit expressions for the operators $Q_0, Q_2, \ldots, Q_{4g-2}$ for $g = 1,2,3,4$.

\text{}
\end{abstract}

\maketitle

\rightline{\it Dedicated to the memory of remarkable mathematician}
\rightline{\it Viktor Zelikovich Enolskii (1945 - 2019),}
\rightline{\it who made a great contribution to the development and applications}
\rightline{\it of the theory of multidimensional sigma functions.}

\text{}

Viktor Enolskii's contribution to mathematics, mathematical and theoretical physics, including his results on the theory and applications of multidimensional sigma functions, is described in a memorial article to appear in the Notices of AMS.

\text{}

\text{}

\section{Introduction} 

The heat equation, equivalent to Schr\"odinger equation, plays a fundamental role in the theory of Abelian functions. For example, part I of the monograph \cite{Mum} begins with characterizing the elliptic theta function as a periodic fundamental solution of the one-dimensional heat equation.

 Directly from the construction of the function $\theta(z, \Omega)$, where $z \in \mathbb{C}^g$, and $\Omega$ is a symmetric $(g \times g)$ matrix, follows
 that this function satisfies the classical system of multidimensional heat equations.
In this system the times are the $g(g+1)/2$ independant parameters of the manifold of symmetric $(g \times g)$ matrices $\Omega$. For this system, the vector fields along the times form a holonomic frame.

For $g>3$ the dimension of the manifold of symmetric $(g \times g)$ matrices that are period matrices of a non-singular Riemann surface of genus $g$ is less than $g(g+1)/2$. Here arises the well-known Riemann-Schottky problem, which has been solved by T.\,Shiota (see \cite{Sh}) based on a hypothesis by S.\,P.\,Novikov, in terms of the~Kadomtsev--Petviashvili equation.

Within the framework of F. Klein's problem, for each non-singular Riemann surface~$V$ of genus $g$, an entire function on $\mathbb{C}^g$ is constructed, namely, the multidimensional sigma function. It's logarithmic derivatives of order $2$ and higher
generate
the whole field of meromorphic functions on the Jacobian of the curve $V$ (see \cite{BEL-97}).
In it's first constructions, this multidimensional sigma function was introduced as a modification of the multidimensional theta function.
The efficiency of its implementation is achieved through the specially selected choice of the curve model.
This results in an entire function whose series expansion in the vector argument $z$ in $\mathbb{C}^g$ gives a series over the ring of polynomials in the parameters of the equation of the curve.
By choosing a family of curves (hyperelliptic or a wider family of $(n, s)$-curves), we can represent theta functions of these curves as modified sigma functions of these curves.
An effective description of such modifications gives a solution to the Riemann--Schottky problem that is special for the family of curves under consideration.
But, in this case, we cannot use the classical system of heat equations to characterize the obtained theta functions, since the vector fields along the times corresponding to the variations of the periods can intersect the discriminants of the curves. This is again connected to the Riemann--Schottky problem.  

In \cite{Nonhol} was constructed a system of heat equations that characterizes the multidimensional sigma functions. In this system the vector fields along the times corresponding to the variations of the parameters of the equations of the curves, form a non-holonomic frame. The Schr\"odinger equations  in magnetic fields arise.

Naturally arises the question of an effective description of this system of equations. This work is devoted to the solution of this question in the case of a family of hyperelliptic curves.

\vfill

\eject

\section{The problem in the case of a family of hyperelliptic curves} 

We consider hyperelliptic curves of genus $g \in \mathbb{N}$ in the model
\begin{equation} \label{mod}
\mathcal{V}_\lambda = \{(x, y)\in\mathbb{C}^2 \colon
y^2 = x^{2g+1} + \lambda_4 x^{2 g - 1}  + \lambda_6 x^{2 g - 2} + \ldots + \lambda_{4 g} x + \lambda_{4 g + 2}\}. 
\end{equation}
The curve depends on the parameters $\lambda = (\lambda_4, \lambda_6, \ldots, \lambda_{4 g}, \lambda_{4 g + 2}) \in \mathbb{C}^{2 g}$.
Let $\mathcal{B} \subset \mathbb{C}^{2g}$ be the subspace of parameters such that the curve $\mathcal{V}_{\lambda}$ is nonsingular for~$\lambda \in~\mathcal{B}$.
Then we have $\mathcal{B} = \mathbb{C}^{2g} \backslash \Sigma$, where $\Sigma$ is~the discriminant hypersurface of the~universal curve.

For a meromorphic function $f$ in $\mathbb{C}^g$ the vector $\omega \in \mathbb{C}^g$ is a period if $f(z+\omega) = f(z)$ for all~$z \in \mathbb{C}^g$.
If a meromorphic function $f$
has $2g$ independent periods in~$\mathbb{C}^g$, then $f$ is called an \emph{abelian function}.
Therefore an abelian function is a  meromorphic~function on the complex torus $T^g = \mathbb{C}^g/\Gamma$,
where $\Gamma$ is the lattice formed by the periods.

For each $\lambda \in \mathcal{B}$ the set of periods of holomorphic differentials on the curve $\mathcal{V}_\lambda$ generates a lattice $\Gamma_\lambda$ of rank $2 g$ in $\mathbb{C}^g$. 
A \emph{hyperelliptic function of genus} $g$ (see \cite{B2, BEL-12, BEL18}) is a~meromorphic function in $\mathbb{C}^g \times \mathcal{B}$,
such that for each $\lambda \in \mathcal{B}$ it's restriction on $\mathbb{C}^g \times \lambda$
is an abelian function, where the torus~$T^g$ is~the~Jacobian $\mathcal{J}_\lambda = \mathbb{C}^g/\Gamma_\lambda$ of the curve~$\mathcal{V}_\lambda$.
We denote by $\mathcal{F}$ the field of hyperelliptic functions of genus $g$. For the properties of this field, see \cite{BEL-12, BEL18}. 
 
We denote the coordinates in~$\mathbb{C}^g$ by~$z = (z_1, z_3, \ldots, z_{2g-1})$. The indices of the coordinates $z = (z_1, z_3, \ldots, z_{2g-1}) \in \mathbb{C}^g$ and of the parameters $\lambda = (\lambda_4, \lambda_6, \ldots, \lambda_{4 g}, \lambda_{4 g + 2}) \in \mathbb{C}^{2 g}$ determine their weights:
$\wt z_k = - k$, $\wt \lambda_k = k$.
We denote by $P$ the ring of polynomials in $\lambda \in \mathcal{B} \subset \mathbb{C}^{2g}$.

We consider polynomial Lie algebras \cite{BPol} of vector fields tangent to the discriminant $\Sigma$ in~$\mathbb{C}^{2g}$. Their generators $\{L_{0}, L_{2}, L_{4}, \ldots, L_{4 g - 2}\}$ are the vector fields 
\[
 L_{2k} = \sum_{s=2}^{2g+1} v_{2k+2, 2s-2}(\lambda) {\partial \over \partial \lambda_{2k}}, \quad \text{where} \quad v_{2k+2, 2s-2}(\lambda) \in P.
\]
At a point $\lambda \in \mathcal{B}$ this vector fields determine a $2g$-dimensional nonholonomic frame.
The structure of a Lie algebra in the $P$-module with the
generators
$
1, \; L_{0}, \; L_{2}, \; L_{4}, \; \ldots, \; L_{4 g - 2}
$ is determined by the polynomial matrices $V(\lambda)=(v_{2i,2j}(\lambda)),$ where
$i,j=1,\dots,2g$, and~$C(\lambda) = \{c_{2i,2j}^{2k}(\lambda)\}$, where $i,j,k=0,\dots,2g-1$, such that 
\begin{equation}
[L_{2i},L_{2j}]=\sum_{k=0}^{2g-1} c_{2i,2j}^{2k}(\lambda)L_{2k},
\quad[L_{2i},\lambda_{2q}]=v_{2i+2,2q-2}(\lambda),
\quad[\lambda_{2q},\lambda_{2r}]=0\label{al}.
\end{equation}
Here $\lambda_q$ is the operator of multiplication 
by
the function $\lambda_q$ in $P$.

We can find the explicit expressions for the matrix $V(\lambda)$ in \cite[\S 4.1]{A} (see also \cite{ BPol} and  \cite[Lemma~3.1]{4A}). For~convenience, we~assume that $\lambda_s = 0$ for all $s \notin \{0,4,6, \ldots, 4 g, 4 g + 2\}$ and $\lambda_0 = 1$.
For $k, m \in \{ 1, 2, \ldots, 2 g\}$, $k\leqslant m$, we set
\[
 v_{2k, 2m}(\lambda) = \sum_{s=0}^{k-1} 2 (k + m - 2 s) \lambda_{2s} \lambda_{2 (k+m-s)}
 - {2 k (2 g - m + 1) \over 2 g + 1} \lambda_{2k} \lambda_{2m},
\] 
and for $k > m $ we set $v_{2k, 2m}(\lambda) = v_{2m, 2k}(\lambda)$.

The vector field $L_0$ is the Euler vector field, namely, since $\wt \lambda_{2k} = 2 k$, we have
\begin{align} \label{L0}
 [L_0, \lambda_{2k}] &= 2 k \lambda_{2k}, & [L_0, L_{2k}] &= 2 k L_{2k}.
\end{align}
This determines the weights of the vector fields $L_k$, namely $\wt L_{2k} = 2 k$. 
The structure of a Lie algebra described above gives  in this case a graded polynomial Lie algebra \cite{BPol} that we denote by $\mathscr{L}_{L}$.
The structure polynomials $c_{2i,2j}^{2s}(\lambda)$ are described in~\cite[Theorem~2.5]{Nonhol}.

There is the well-known Lie--Witt algebra $W_{\geqslant}$ over the field of complex numbers $\mathbb{C}$ generated by the operators $l_{2i}$, where $i = 0, 1, 2, \ldots$, with the commutation relations
\[
 [l_{2i}, l_{2j}] = 2 (j-i) l_{2 (i+j)}.
\]
With respect to the bracket $[\cdot,\cdot]$ the  Lie--Witt algebra $W_{\geqslant}$ is generated by the three operators $l_0$, $l_2$, $l_4$.
The graded polynomial Lie algebra $\mathscr{L}_{L}$
over $P$ is a deformation of the Lie--Witt algebra $W_{\geqslant}$. It is as well generated by only three operators  $L_0$, $L_2$, and $L_4$. The relation holds (see Lemma \ref{L33}):
\[
[L_2, L_{2k}] = 2 (k-1) L_{2k+2} + {4 (2 g - k) \over (2 g + 1)} \left( \lambda_{2k+2} L_0 - \lambda_4 L_{2k-2}\right). 
\]

Now we introduce the Schr\"odinger operators. We consider the space $\mathbb{C}^{3g}$ with coordinates $(z, \lambda)$.
We denote by $\CU$ the ring of differentiable functions in $z$ and
$\lambda$. We set
\begin{equation} \label{HQ}
Q_{2k} = L_{2k} - H_{2k}, \qquad k = 0, 1, 2, \ldots, 2g-1,
\end{equation}
where
\begin{equation} \label{e2}
H_{2k} = {1 \over 2} \sum \left( \alpha_{a,b}^{(k)}(\lambda)\partial_a \partial_b + 2 \beta_{a,b}^{(k)}(\lambda)z_a\partial_b + \gamma_{a,b}^{(k)}(\lambda)z_a z_b\right) + \delta^{(k)}(\lambda),
\end{equation}
and the summation is over odd $a, b$ from $1$ to $2g-1$. In \cite{Nonhol} a solution to the following problem is given: 

\begin{prb}\label{pro1} Find the sufficient conditions on the 
data
$\bigl\{\alpha^{(i)}(\lambda),\beta^{(i)}(\lambda),\gamma^{(i)}(\lambda),\delta^{(i)}(\lambda)\bigr\}$
for the operators \eqref{HQ} to give a representation of Lie
algebra \eqref{al} in the ring of operators on $\CU.$
\end{prb}

\begin{dfn}
The system of equations for $\varphi = \varphi(z, \lambda)$
\begin{equation} \label{e3}
Q_{2k} \varphi = 0
\end{equation}
is called the \emph{system of heat equations}. The operators $Q_{2k}$ are called \emph{Schr\"odinger operators}.
\end{dfn}

We use the theory of hyperelliptic Kleinian functions (see \cite{BEL-12, Baker, BEL, BEL-97}, and~\cite{WW} for elliptic functions).
Take the coordinates
$(z, \lambda)$
in $\mathbb{C}^g \times \mathcal{B} \subset \mathbb{C}^{3g}$.
Let $\sigma(z, \lambda)$ be the hyperelliptic sigma function (or elliptic sigma function in genus $g=1$ case). We denote $\partial_k = {\partial \over \partial z_k}$.
Following \cite{B2, BEL18, B3}, we use the notation
\begin{equation} \label{n}
\zeta_{k} = \partial_k \ln \sigma(z, \lambda), \qquad
\wp_{k_1, \ldots, k_n} = - \partial_{k_1} \cdots \partial_{k_n} \ln \sigma(z, \lambda),
\end{equation}
where $n \geqslant 2$, $k_s \in \{ 1, 3, \ldots, 2 g - 1\}$.
The functions $\wp_{k_1, \ldots, k_n}$ provide us with examples of hyperelliptic functions. 
The field $\mathcal{F}$ is the field of fractions of the ring of polynomials $\mathcal{P}$ generated by the functions $\wp_{k_1, \ldots, k_n}$, where $n \geqslant 2$, $k_s \in \{ 1, 3, \ldots, 2 g - 1\}$.

As shown in \cite{Nonhol}, a system of heat equations \eqref{e3} for operators $Q_{2k}$ that give a solution to~Problem 1.1 determines the hyperelliptic sigma function $\sigma(z, \lambda)$, which allows to construct the hyperelliptic Kleinian functions theory starting from such operators.

The construction of the operators $Q_{2i}$ in \cite{Nonhol} uses the condition \cite[equation (1.3)]{Nonhol} stating that the commutator of operators $[Q_{2i},Q_{2j}]$ is determined by a formula over $P$ with the same coefficients as the formula for $[L_{2i}, L_{2j}]$, 
namely, the polynomial algebra generated by the operators $Q_{2i}$ with $i=0, 1, ...$ is yet another realization of Lie--Witt algebra $W_{\geqslant}$ deformation. 
Therefore, for an effective description of polynomial Lie algebra $\mathscr{L}_{Q}$ one needs to obtain explicit formulas for $Q_0, Q_2$, and $Q_4$.
These formulas are the main result of this work.
As an application, we give the explicit form for differential operators in the case of universal hyperelliptic curve of genus $4$.

\vfill

\eject

\section{Generating functions for Schr\"odinger operators} \label{S2}

In this section, in the case of the model~\eqref{mod} we present the explicit solution to Problem~\ref{pro1} from \cite[Theorem 2.6]{Nonhol}.
We will shift the original notation to be consistent with the formulas introduced in this work.

We denote by $q(a(x), b(x))$ the quotient in the Euclidean division of the polynomial $a(x)$ by the polynomial $b(x)$ and by $r(a(x), b(x))$ the residue in the Euclidean division of the polynomial $a(x)$ by the polynomial $b(x)$. We set (see \eqref{mod})
\[
 f(x) = x^{2g+1} + \sum_{k=0}^{2g-1} \lambda_{2 (2g+1-k)} x^k
\]
and define the generating functions for the operators $L_{2k}$ and $H_{2k}$:
\begin{align} \label{gf}
L(x) &= x^{2g-1} \sum_{k=0}^{2g-1} x^{-k} L_{2k}, &
H(x) &= x^{2g-1} \sum_{k=0}^{2g-1} x^{-k} H_{2k}.
\end{align}
For $R_i(x) = x^{g-i+1} \partial_x q(f(x), x^{2g-2i+2})$ we set
\begin{align*}
h(x) &= \sum_{i=1}^g x^{g-i} \partial_{2i-1} + R_{i}(x) z_{2i-1},\\
t(x) &= {1 \over 2} \sum_{i = 1}^{g} (g-i+1) z_{2i-1}^2 q(R_{i}(x), x^{g-i+2}) + \\
& + \sum_{i = 1}^{g-1} \sum_{j=i+1}^{g} (g-j+1) z_{2j-1} q(x^{g-i} \partial_{2i-1} + R_{i}(x) z_{2i-1}, x^{g-j+2}).
\end{align*}
Then we have the relation (see \cite[Theorem 2.6]{Nonhol}):
\begin{equation} \label{Qx}
H(x) = r\left(- {1 \over 4} f''(x) + 2 f(x) t(x) + {1 \over 2} h(x) \circ h(x), f'(x)\right), 
\end{equation}
where $\circ$ denotes the composition of operators.

The function $Q(x) = L(x) - H(x)$ is the generating function for Shr\"odinger operators.

\begin{lem} \label{L21}
For Schr\"odinger operators, in~\eqref{e2} we have
\begin{align*}
&\alpha_{a,b}^{(k)}(\lambda) = 1, \quad \text{if} \quad a+b = 2 k, \quad \text{and} \quad  a, b \in 2 \mathbb{N} + 1,\\
&\alpha_{a,b}^{(k)}(\lambda) = 0, \quad \text{if} \quad a+b \ne 2 k, \quad \text{and} \quad  a, b \in 2 \mathbb{N} + 1,\\
&\delta^{(k)}(\lambda) = \left(- {1 \over 4} (2g-k+1) (2g-k) + {1 \over 2} \left( g + \left[ {k+1 \over 2} \right] - k \right) \left( g - \left[ {k+1 \over 2} \right]\right) \right) \lambda_{2k}.
\end{align*}
\end{lem}

\begin{proof}
The coefficient at $\partial_a \partial_b$ in $H(x)$ comes from the summand ${1 \over 2} h(x) \circ h(x)$, by expanding we get $\sum_{i=1}^g \sum_{j=1}^g x^{2g-i-j} \partial_{2i-1} \partial_{2j-1}$. Thus from \eqref{gf} we get the expressions for~$\alpha_{a,b}^{(k)}(\lambda)$.

The expression for $\delta^{(k)}(\lambda)$ is obtained from the summands $- {1 \over 4} f''(x)$ and ${1 \over 2} h(x) \circ h(x)$. The first gives $- {1 \over 4} (2g-k+1) (2g-k) \lambda_{2k}$ and the second gives ${1 \over 2} (g+\left[{k+1 \over 2}\right] - k) (g - \left[{k+1 \over 2}\right]) \lambda_{2k}$, 
which results in the answer.
\end{proof}

\begin{rem}
We will later show that for Schr\"odinger operators, in~\eqref{e2} we have
\begin{align*}
&\beta_{a,b}^{(k)}(\lambda) \text{ is a linear function in } \lambda,\\
&\gamma_{a,b}^{(k)}(\lambda) \text{ is a quadratic function in } \lambda.
\end{align*}
See Lemma \ref{L22}.
\end{rem}

\section{Explicit form for Schr\"odinger operators} \label{S3}

Recall that in our notation~$\lambda_s = 0$ for all $s \notin \{0,4,6, \ldots, 4 g, 4 g + 2\}$ and $\lambda_0 = 1$.

\begin{thm} \label{T31}
We have the explicit expressions:
\begin{align*}
H_0 &= \sum_{s=1}^g (2s-1) z_{2s-1} \partial_{2s-1} - {g (g+1) \over 2};\\
H_2 &= {1 \over 2} \partial_1^2 + \sum_{s=1}^{g-1} (2s-1) z_{2s-1} \partial_{2s+1}
- {4 \over 2 g + 1} \lambda_4  \sum_{s=1}^{g-1} (g - s)  z_{2s+1} \partial_{2s-1}
+ \\ & \qquad \qquad +\sum_{s=1}^{g} \left({2s - 1 \over 2} \lambda_{4s} - {2 (g - s + 1)  \over 2 g+1} \lambda_{4} \lambda_{4s - 4} \right) z_{2s-1}^2;\\
H_4 &= \partial_1 \partial_3 + \sum_{s=1}^{g-2} (2s-1) z_{2s-1} \partial_{2s+3}
+ \lambda_4 \sum_{s=1}^{g-1} (2s-1) z_{2s+1} \partial_{2s+1} - \\
& \qquad \qquad - {6 \over 2 g + 1} \lambda_6 \sum_{s=1}^{g-1} (g - s) z_{2s+1} \partial_{2s-1} +\\
& \qquad \qquad + \sum_{s=1}^{g}\left( (2 s - 1) \lambda_{4 s +2} - {3 (g - s + 1) \over 2 g + 1} \lambda_6 \lambda_{4 s - 4} \right) z_{2 s-1}^2 +\\
& \qquad \qquad + \sum_{s=1}^{g-1} (2 s - 1) \lambda_{4s+4} z_{2s-1} z_{2s+1} - {g (g-1) \over 2} \lambda_4.
\end{align*}
\end{thm}
\vspace{-13pt}
\begin{proof}
The explicit expressions for
$\alpha_{a,b}^{(k)}(\lambda)$ and $\delta^{(k)}(\lambda)$ are obtained in Lemma \ref{L21}. We see that for $k = 0,1,2$ they coincide with the ones given in the Theorem. Recall $\lambda_0 = 1$, $\lambda_2 = 0$.

In \eqref{Qx} the coefficient at $z_{2j-1} \partial_{2i-1}$ is equal to 
\[
r(2 (g-j+1) f(x) q(x^{g-i}, x^{g-j+2}) + x^{g-i} R_{j}(x), f'(x)).
\]
For $H_0$ we are interested in the coefficient at $x^{2g-1}$ of this polynomial, for $H_2$ in the coefficient at $x^{2g-2}$ and for $H_4$ in the coefficient at $x^{2g-3}$.

For $j < i+1$ the polynomial is equal to
\[
r(x^{g-i} R_{j}(x), f'(x)) = x^{g-i} R_{j}(x)
\]
because the degree of the polynomial $R_{j}(x)$ in $x$ is $g+j-1$ and the degree of $f'(x)$ is $2g$. Therefore we get the coefficient $(2j-1)$ in $j=i$ for $H_0$, in $j=i-1$ for $H_2$, and in $j=i-2$ for $H_4$, as well as the coefficient $(2 j - 3) \lambda_4$  in $j=i \geqslant 2$ for $H_4$, all the other coefficients being zero. 

For $j = i+1$ the polynomial is equal to
\begin{multline*}
r(x^{g-i} R_{j}(x), f'(x)) = r(x^{2g-2i} \partial_x q(f(x), x^{2g-2i}), f'(x)) = \\ = r( (2i+1) x^{2g}, f'(x)) +  \sum_{k=2g-2i+1}^{2g-1} (k-2g+2i) \lambda_{2 (2g+1-k)} x^{k-1} = \\ = - {2i+1 \over 2g+1} \sum_{k=0}^{2g-1} k \lambda_{2 (2g+1-k)} x^{k-1} +  \sum_{k=2g-2i+1}^{2g-1} (k-2g+2i) \lambda_{2 (2g+1-k)} x^{k-1}.
\end{multline*}
Therefore the coefficient for $H_0$ is zero, the coefficient for $H_2$ is $ - {4 \over 2g+1} (g - i)  \lambda_{4}$, the coefficient for $H_4$ is $- {6 \over 2g+1} (g - i) \lambda_{6}$.

For $j \geqslant i+2$ the polynomial is equal to
\begin{multline*}
 r(2 (g-j+1) x^{j-i-2} f(x) + x^{2g-i-j+1} \partial_x q(f(x), x^{2g-2j+2}), f'(x)) = \\
 =  r(f'(x) x^{j-i+1} + \sum_{k=0}^{2g-2j+2} (2g-2j+2-k) \lambda_{2(2g+1-k)} x^{k+j-i-2}, f'(x)) = \\ = \sum_{k=0}^{2g-2j+2} (2g-2j+2-k) \lambda_{2(2g+1-k)} x^{k+j-i-2}.
\end{multline*}
As $j \geqslant i+2 \geqslant 3$,
the coefficients for $H_0$, $H_2$ and $H_4$ are zero.

In \eqref{Qx} the coefficient at $z_{2i-1}^2$ is equal to
\begin{equation} \label{s}
r\left((g-i+1) f(x) q(R_{i}(x), x^{g-i+2}) + {1 \over 2} R_i(x)^2, f'(x)\right).
\end{equation}

For $i=1$ we note that $R_1(x) = x^{g}$ and the expression \eqref{s} takes the form
\[
{1 \over 2} r\left(x^{2g}, f'(x)\right) = - {1 \over 2 (2g+1)} \sum_{k=1}^{2g-1} k \lambda_{2 (2g+1-k)} x^{k-1}.
\]
In this expression the coefficient at $x^{2g-1}$ is zero, the coefficients at $x^{2g-2}$ and $x^{2g-3}$ provide  the corresponding coefficients for $H_2$ and $H_4$ at $z_1^2$.

For $i > 1$ we have
\begin{multline*}
r\left((g-i+1) f(x) q(R_{i}(x), x^{g-i+2}) + {1 \over 2} R_i(x)^2, f'(x)\right) = \\
= r\Big({1 \over 2} R_i(x) x^{-g+i-1} \left(2 (g-i+1) x^{-1} f(x) + x^{g-i+1} R_i(x)\right) - \\ - (g-i+1) f(x) x^{-1} r(x^{-g+i-1} R_{i}(x), x), f'(x)\Big) = \\
= r\Big({1 \over 2} R_i(x) x^{-g+i-1} \left(f'(x) + x^{-1} \sum_{k=0}^{2g-2i+1} (2g-2i+2-k) \lambda_{2 (2g+1-k)} x^k\right) - \\ - (g-i+1) f(x) x^{-1} r(x^{-g+i-1} R_{i}(x), x), f'(x)\Big) = \\
= {1 \over 2} R_i(x) x^{-g+i-1} \left(x^{-1} \sum_{k=0}^{2g-2i+1} (2g-2i+2-k) \lambda_{2 (2g+1-k)} x^k\right) - \\ - (g-i+1) r(x^{-g+i-1} R_{i}(x), x) x^{-1} (f(x) - x^{2g+1}) - \\
- (g-i+1) r(x^{-g+i-1} R_{i}(x), x) r\Big(x^{2g}, f'(x)\Big) =
\\
= {1 \over 2} R_i(x) x^{-g+i-1}
\left(x^{-1} \sum_{k=0}^{2g-2i+1} (2g-2i+2-k) \lambda_{2 (2g+1-k)} x^k\right) - \\ - (g-i+1) \lambda_{4 (i-1)} x^{-1} \left(\sum_{k=0}^{2g-1} \lambda_{2 (2g+1-k)} x^k\right)
+ {(g-i+1) \over (2g+1)} \lambda_{4 (i-1)} \left(\sum_{k=1}^{2g-1} k \lambda_{2 (2g+1-k)} x^{k-1}\right).
\end{multline*}
This is a polynomial of degree $2g-2$. Therefore the coefficient for $H_0$ is zero. The coefficients for $H_2$ and $H_4$ coincide with the corresponding coefficient given in the Theorem.

In \eqref{Qx} the coefficient at $z_{2i-1} z_{2j-1}$, $j > i$, is equal to
\begin{multline*}
r\left(R_i(x) x^{-g+j-2} \left(2 (g-j+1) f(x) + x^{g-j+2} R_j(x)\right), f'(x)\right) = \\
r\left(R_i(x) x^{-g+j-2} \left(x f'(x) + \sum_{k=0}^{2g-2j+1} (2g-2j+2-k) \lambda_{2 (2g+1-k)} x^k\right), f'(x)\right) = \\ =
R_i(x) x^{-g+j-2} \left(\sum_{k=0}^{2g-2j+1} (2g-2j+2-k) \lambda_{2 (2g+1-k)} x^k\right).
\end{multline*}
The degree of this polynomial is equal to $2g+i-j-2\leqslant2g-3$, and the equality holds for $j=i+1$.
Therefore the coefficients for $H_0$ and $H_2$ are zero, and the coefficient for~$H_4$ for $j=i+1$ is
$(2i-1) \lambda_{4i+4}$.
\end{proof}

\begin{lem} \label{L22}
For Schr\"odinger operators, in~\eqref{e2} we have
\begin{align*}
&\beta_{a,b}^{(k)}(\lambda) \text{ is a linear function in } \lambda,\\
&\gamma_{a,b}^{(k)}(\lambda) \text{ is a quadratic function in } \lambda.
\end{align*}
\end{lem}
\textit{The proof} follows from the explicit expressions in the proof of Theorem \ref{T31}.

\begin{lem} \label{L33} The relation holds:
\begin{equation} \label{that}
[L_2, L_{2k}] = 2 (k-1) L_{2k+2} + {4 (2 g - k) \over (2 g + 1)} \left( \lambda_{2k+2} L_0 - \lambda_4 L_{2k-2}\right). 
\end{equation}
\end{lem}
\text{}
\vspace{-16pt}

\begin{proof}
We have
\begin{align*}
 L_{2k} &= 2 \sum_{m = 2}^{k+1} \left( \sum_{s=0}^{m-2} (k + m - 2 s) \lambda_{2s} \lambda_{2 (k+m-s)}
 - {(m - 1) (2 g - k) \over 2 g + 1} \lambda_{2(k+1)} \lambda_{2(m-1)} \right) {\partial \over \partial \lambda_{2m}} + \\
 & + 2 \sum_{m = k+2}^{2 g +1} \left( \sum_{s=0}^{k} (k + m - 2 s) \lambda_{2s} \lambda_{2 (k+m-s)}
 - {(k+1) (2 g - m +2) \over 2 g + 1} \lambda_{2(k+1)} \lambda_{2(m-1)} \right) {\partial \over \partial \lambda_{2m}};\\
 L_{2} &=  \sum_{m = 2}^{2 g +1} \left(2 (m+1) \lambda_{2 (m+1)}
 - {4 (2 g - m +2) \over 2 g + 1} \lambda_{4} \lambda_{2(m-1)} \right) {\partial \over \partial \lambda_{2m}}.
\end{align*} 
Now we obtain \eqref{that} by a direct calculation of the coefficients.
\end{proof}

\begin{cor} \label{t345}
For $k = 3, 4, 5, \ldots, 2g-1$ the formula holds
\begin{equation} \label{e12}
Q_{2k} = {1 \over 2 (k-2)} [Q_2, Q_{2k-2}] - {2 (2 g - k + 1) \over (k-2) (2 g + 1)} \left( \lambda_{2k} Q_0 - \lambda_4 Q_{2k-4}\right). 
\end{equation}
\end{cor}

This determines the operators $Q_{2k}$ for $k = 3, 4, 5, \ldots, 2g-1$ recurrently and explicitly.

\begin{proof}
 The operators $Q_{2k}$ by construction (see Section \ref{S2}) give a solution to Problem~\ref{pro1}. The expression \eqref{e12} is equivalent to \eqref{that} rewritten in $Q_{2k}$.
\end{proof}

\vfill

\eject

\section{Explicit formulas for Schr\"odinger operators in nonholonomic frame} \label{S4}

The following operators can be found in \cite{Nonhol, B2, BL0, BL} for genus $g=1$, in \cite{B2, BL0, BL} for genus $g=2$, in~\cite{4A, BB19} for genus $g=3$. The genus $g=4$ case is new. It follows from the formulas given in Section \ref{S3}.

\subsection{Schr\"odinger operators for genus $g=1$} \text{}

In this case, the explicit formulas for $\{H_{2k}\}$ in \eqref{HQ} are
\begin{align*}
H_0 &= z_1 \partial_1 - 1; &
H_2 &= {1 \over 2} \partial_1^2 - {1 \over 6} \lambda_4 z_1^2.
\end{align*}

\subsection{Schr\"odinger operators for genus $g=2$} \text{}

In this case, the explicit formulas for $\{H_{2k}\}$ in \eqref{HQ} are
\begin{align*}
H_0 &= z_1 \partial_1 + 3 z_3 \partial_3 - 3;\\
H_2 &= {1 \over 2} \partial_1^2 - {4 \over 5} \lambda_4 z_3 \partial_1 + z_1 \partial_3 - {3 \over 10} \lambda_4 z_1^2 + \left({3 \over 2} \lambda_8 - {2 \over 5} \lambda_4^2\right) z_3^2;\\
H_4 &= \partial_1 \partial_3 - {6 \over 5} \lambda_6 z_3 \partial_1 + \lambda_4 z_3 \partial_3 - {1 \over 5} \lambda_6 z_1^2 + \lambda_8 z_1 z_3 + \left(3 \lambda_{10} - {3 \over 5} \lambda_4 \lambda_6\right) z_3^2 - \lambda_4;\\
H_6 &= {1 \over 2} \partial_3^2 - {3 \over 5} \lambda_8 z_3 \partial_1 - {1 \over 10} \lambda_8 z_1^2 + 2 \lambda_{10} z_1 z_3 - {3 \over 10} \lambda_4 \lambda_8 z_3^2 - {1 \over 2} \lambda_6.
\end{align*}

\subsection{Schr\"odinger operators for genus $g=3$} \text{}

In this case, the explicit formulas for $\{H_{2k}\}$ in \eqref{HQ} are 
\begin{align*}
H_0 &=
z_1 \partial_1 + 3 z_3 \partial_3 + 5 z_5 \partial_5 - 6;
\\
H_2 &= {1 \over 2} \partial_1^2 - {8 \over 7} \lambda_4 z_3 \partial_1 + \left(z_1 - {4 \over 7} \lambda_4 z_5\right) \partial_3 + 3 z_3 \partial_5 + \\
& - {5 \over 14} \lambda_4 z_1^2 + \left({3 \over 2} \lambda_8 - {4 \over 7} \lambda_4^2\right) z_3^2 +
 \left({5 \over 2} \lambda_{12} - {2 \over 7} \lambda_4 \lambda_8 \right) z_5^2;\\
H_4 &= \partial_1 \partial_3 - {12 \over 7} \lambda_6 z_3 \partial_1 + \left(\lambda_4 z_3 - {6 \over 7} \lambda_6 z_5\right) \partial_3 + \left(z_1 + 3 \lambda_4 z_5 \right) \partial_5 - {2 \over 7} \lambda_6 z_1^2 + \\
& + \lambda_8 z_1 z_3 
 + \left(3 \lambda_{10} - {6 \over 7} \lambda_4 \lambda_6\right) z_3^2 + 3 \lambda_{12} z_3 z_5 + \left(5  \lambda_{14} - {3 \over 7} \lambda_6 \lambda_8 \right) z_5^2
  - 3 \lambda_4;
\\
H_6 &= {1 \over 2} \partial_3^2 + \partial_1 \partial_5 - {9 \over 7} \lambda_8 z_3 \partial_1 - {8 \over 7} \lambda_8 z_5 \partial_3 + \left(\lambda_4 z_3 + 2 \lambda_6 z_5\right) \partial_5 - {3 \over 14} \lambda_8 z_1^2 + 2 \lambda_{10} z_1 z_3 + \\ &  + \left({9 \over 2} \lambda_{12} - {9 \over 14} \lambda_4 \lambda_8\right) z_3^2 + \lambda_{12} z_1 z_5 + 6 \lambda_{14} z_3 z_5 + \left({3 \over 2} \lambda_4 \lambda_{12} - {4 \over 7} \lambda_8^2\right) z_5^2 - 2 \lambda_6;\\
H_8 &= \partial_3 \partial_5 - \left({6 \over 7} \lambda_{10} z_3 - \lambda_{12} z_5\right) \partial_1 - {10 \over 7} \lambda_{10} z_5 \partial_3 + \lambda_8 z_5 \partial_5 - {1 \over 7} \lambda_{10} z_1^2 + 3 \lambda_{12} z_1 z_3 +\\+ &\left( 6 \lambda_{14} - {3 \over 7} \lambda_4 \lambda_{10}\right) z_3^2 + 2 \lambda_{14} z_1 z_5 + \lambda_4 \lambda_{12} z_3 z_5 + \left(3 \lambda_{4} \lambda_{14} + \lambda_6 \lambda_{12} -{5 \over 7} \lambda_8 \lambda_{10}\right) z_5^2 - \lambda_8;\\
H_{10} &= {1 \over 2} \partial_5^2 - \left( {3 \over 7} \lambda_{12} z_3 - 2 \lambda_{14} z_5 \right) \partial_1 - {5 \over 7} \lambda_{12} z_5 \partial_3 - \\ & -{1 \over 14} \lambda_{12} z_1^2 + 4 \lambda_{14} z_1 z_3 - {3 \over 14} \lambda_4 \lambda_{12} z_3^2 + 2 \lambda_4 \lambda_{14} z_3 z_5 + \left(2 \lambda_6 \lambda_{14} - {5 \over 14} \lambda_8 \lambda_{12}\right) z_5^2  - {1 \over 2} \lambda_{10}.
\end{align*}

\subsection{Schr\"odinger operators for genus $g=4$} \text{}

In this case, the explicit formulas for $\{H_{2k}\}$ in \eqref{HQ} are 

\begin{align*}
H_0 &=
z_1 \partial_1 + 3 z_3 \partial_3 + 5 z_5 \partial_5 + 7 z_7 \partial_7 - 10;\\
H_2 &= {1 \over 2} \partial_1^2 + z_1 \partial_3 + 3 z_3 \partial_5 + 5 z_5 \partial_7 - {4 \over 9} \lambda_4 \left(3 z_{3} \partial_{1} + 2 z_{5} \partial_{3} + z_{7} \partial_{5} \right) - \\
& \qquad - {7 \over 18} \lambda_4 z_1^2 + \left({3 \over 2} \lambda_8 - {2 \over 3} \lambda_4^2 \right) z_3^2 + \left( {5 \over 2} \lambda_{12} - {4 \over 9} \lambda_4 \lambda_8 \right) z_5^2 + \left({7 \over 2} \lambda_{16} - {2 \over 9} \lambda_4 \lambda_{12}  \right) z_7^2;\\
H_4 &= \partial_1 \partial_3 + z_1 \partial_5 + 3 z_3 \partial_7 +
\lambda_4 \left( z_3 \partial_3 + 3 z_5 \partial_5 + 5 z_7 \partial_7 \right)
- {2 \over 3} \lambda_6 \left( 3 z_3 \partial_1 + 2 z_5 \partial_3 + z_7 \partial_5 \right)-  \\
& \qquad - {1 \over 3} \lambda_6 z_1^2 + \lambda_8 z_1 z_3 + \left(3 \lambda_{10} - \lambda_4 \lambda_6\right) z_3^2 + 3 \lambda_{12} z_3 z_5 + \left(5 \lambda_{14} - {2 \over 3} \lambda_6 \lambda_8\right) z_5^2 + \\
& \qquad + 5 \lambda_{16} z_5 z_7 + \left(7 \lambda_{18} - {1 \over 3} \lambda_6 \lambda_{12}\right) z_7^2 - 6 \lambda_4;
\\
H_6 &= {1 \over 2} \partial_3^2 + \partial_1 \partial_5 - {5 \over 3} \lambda_8 z_3 \partial_1 - {16 \over 9} \lambda_8 z_5 \partial_3 + \\ & \qquad + \left( \lambda_4 z_3 + 2 \lambda_6 z_5 -{8 \over 9} \lambda_8 z_7 \right) \partial_5 + \left(z_1 +3 \lambda_4 z_5 + 4 \lambda_6 z_7\right) \partial_7 -\\ & \qquad 
- {5 \over 18} \lambda_8 z_1^2 + 2 \lambda_{10} z_1 z_3 + \lambda_{12} z_1 z_5 - \left({5 \over 6} \lambda_4 \lambda_8 - {9 \over 2} \lambda_{12}\right) z_3^2 + 6 \lambda_{14} z_3 z_5 + 3 \lambda_{16} z_3 z_7 + \\
& \qquad + \left({3 \over 2} \lambda_4 \lambda_{12} - {8 \over 9} \lambda_8^2 + {15 \over 2} \lambda_{16}\right) z_5^2 + 10 \lambda_{18} z_5 z_7 - \left({4 \over 9} \lambda_8 \lambda_{12} - {5 \over 2} \lambda_4 \lambda_{16}\right) z_7^2
- {9 \over 2} \lambda_6;
\\
H_8 &= \partial_3 \partial_5 + \partial_1 \partial_7 - \left({4 \over 3} \lambda_{10} z_3 - \lambda_{12} z_5\right) \partial_1 - {20 \over 9} \lambda_{10} z_5 \partial_3 + \\ & \qquad + \left(\lambda_8 z_5 - {10 \over 9} \lambda_{10} z_7\right) \partial_5 + \left(\lambda_4 z_3 + 2 \lambda_6 z_5 + 3 \lambda_8 z_7\right) \partial_7 - \\
& \qquad - {2 \over 9} \lambda_{10} z_1^2 + 3 \lambda_{12} z_1 z_3 + 2 \lambda_{14} z_1 z_5 + \lambda_{16} z_1 z_7 - \left({2 \over 3} \lambda_4 \lambda_{10} - 6 \lambda_{14}\right) z_3^2 + \\
& \qquad +\left(\lambda_4 \lambda_{12} + 9 \lambda_{16}\right) z_3 z_5 + 6 \lambda_{18} z_3 z_7 - \left({10 \over 9} \lambda_8 \lambda_{10} - \lambda_6 \lambda_{12} - 3 \lambda_4 \lambda_{14} - 10 \lambda_{18}\right) z_5^2 +\\
& \qquad + 3 \lambda_4 \lambda_{16} z_5 z_7 - \left({5 \over 9} \lambda_{10} \lambda_{12} - 2 \lambda_6 \lambda_{16} - 5 \lambda_4 \lambda_{18}\right) z_7^2 - 3 \lambda_8;
\\
H_{10} &= {1 \over 2} \partial_5^2 + \partial_3 \partial_7
- (\lambda_{12} z_3 - 2 \lambda_{14} z_5 - \lambda_{16} z_7) \partial_1 - {5 \over 3} \lambda_{12} z_5 \partial_3 - \\ & \qquad - {4 \over 3} \lambda_{12} z_7 \partial_5 + (\lambda_8 z_5 + 2 \lambda_{10} z_7) \partial_7 + \\
& \qquad - {1 \over 6} \lambda_{12} z_1^2 + 4 \lambda_{14} z_1 z_3 + 3 \lambda_{16} z_1 z_5 + 2 \lambda_{18} z_1 z_7 - \left({1 \over 2} \lambda_4 \lambda_{12} - {15 \over 2} \lambda_{16} \right) z_3^2 + \\
& \qquad + \left(2 \lambda_4 \lambda_{14} + 12 \lambda_{18}\right) z_3 z_5 + \lambda_4 \lambda_{16} z_3 z_7 - \left({5 \over 6} \lambda_8 \lambda_{12} - 2 \lambda_6 \lambda_{14} - {9 \over 2} \lambda_4 \lambda_{16}\right) z_5^2 + \\
& \qquad + (2 \lambda_6 \lambda_{16} +6 \lambda_4 \lambda_{18}) z_5 z_7 - \left({2 \over 3} \lambda_{12}^2 - {3 \over 2} \lambda_8 \lambda_{16} - 4 \lambda_6 \lambda_{18}\right) z_7^2 - 2 \lambda_{10};
\end{align*}
\vfill 

\eject

\begin{align*}
H_{12} &= \partial_5 \partial_7 - \left({2 \over 3} \lambda_{14} z_3 - 3 \lambda_{16} z_5 - 2 \lambda_{18} z_7\right) \partial_1 - \left({10 \over 9} \lambda_{14} z_5 - \lambda_{16} z_7 \right) \partial_3 - \\ & \qquad - {14 \over 9} \lambda_{14} z_7 \partial_5 + \lambda_{12} z_7 \partial_7 - \\ & \qquad - {1 \over 9} \lambda_{14} z_1^2 + 5 \lambda_{16} z_1 z_3 + 4 \lambda_{18} z_1 z_5 -  \left({1 \over 3} \lambda_4 \lambda_{14} - 9 \lambda_{18}\right) z_3^2 + \\
& \qquad + 3 \lambda_4 \lambda_{16} z_3 z_5 + 2 \lambda_4 \lambda_{18} z_3 z_7 - \left({5 \over 9} \lambda_8 \lambda_{14} - 3 \lambda_6 \lambda_{16} - 6 \lambda_4 \lambda_{18}\right) z_5^2 +  \\
& \qquad + \left(\lambda_8 \lambda_{16} + 4 \lambda_6 \lambda_{18}\right) z_5 z_7 - \left({7 \over 9} \lambda_{12} \lambda_{14} - \lambda_{10} \lambda_{16} - 3 \lambda_8 \lambda_{18}\right) z_7^2 - \lambda_{12};\\
H_{14} &= {1 \over 2} \partial_7^2 - \left({1 \over 3} \lambda_{16} z_3 - 4 \lambda_{18} z_5\right) \partial_1 - \left({5 \over 9} \lambda_{16} z_5 - 2 \lambda_{18} z_7\right) \partial_3 -{7 \over 9} \lambda_{16} z_7 \partial_5 - \\
& \qquad - {1 \over 18} \lambda_{16} z_1^2 + 6 \lambda_{18} z_1 z_3 - {1 \over 6} \lambda_4 \lambda_{16} z_3^2 + 4 \lambda_4 \lambda_{18} z_3 z_5 - \left( {5 \over 18} \lambda_8 \lambda_{16} - 4 \lambda_6 \lambda_{18} \right) z_5^2 + \\
& \qquad + 2 \lambda_8 \lambda_{18} z_5 z_7 - \left({7 \over 18} \lambda_{12} \lambda_{16} - 2 \lambda_{10} \lambda_{18}\right) z_7^2 - {1 \over 2} \lambda_{14}.
\end{align*}

\eject

\section{Application: Differentiation operators for genus $g=4$}

Recall that $\mathcal{F}$ is the field of hyperelliptic functions of genus $g$. In this section we consider the problem of constructing the Lie algebra of derivations of $\mathcal{F}$, i.e. to~find~$3g$ independent differential operators $\mathcal{L}$ such that $\mathcal{L} \mathcal{F} \subset \mathcal{F}$. The exposition to the problem, as well as a general approach to the solution was developed in \cite{BL0, BL}. An~overview is given in \cite{BEL18}. In~\cite{FS, B2, B3} an explicit solution to this problem has been obtained for $g=1,2,3$. Here we give an explicit answer to this problem in the genus $g=4$ case.

We introduce a ring of functions $\mathcal{R}_\varphi$. The generators of this graded ring over $\mathbb{Q}[\lambda]$ are the functions
$\psi_{k_1 \ldots k_n} = - \partial_{k_1} \cdots \partial_{k_n} \ln \varphi$,
where $n \geqslant 2$, $k_s \in \{ 1, 3, \ldots, 2 g - 1\}$, and $\wt \psi_{k_1 \ldots k_n} = k_1 + \ldots + k_n$, $\wt \lambda_k = k$.
We introduce the operators:

\begin{align*}
\mathcal{L}_0 &= L_0
- z_1 \partial_1 - 3 z_3 \partial_3 - 5 z_5 \partial_5 - 7 z_7 \partial_7;
\\
\mathcal{L}_2 &= L_2 
- \psi_1 \partial_1 + {4 \over 3} \lambda_4 z_{3} \partial_1 - \left(z_1 - {8 \over 9} \lambda_4 z_{5}\right) \partial_3 - \left(3 z_3 - {4 \over 9} \lambda_4 z_{7}\right) \partial_5 - 5 z_5 \partial_7;
\\
\mathcal{L}_4 &= L_4
- \psi_3 \partial_1 - \psi_1 \partial_3 - \\ & \quad
+ 2 \lambda_6 z_3 \partial_1
- \left( \lambda_4 z_3
- {4 \over 3} \lambda_6 z_5 \right) \partial_3 
- \left( z_1 + 3 \lambda_4 z_5 
- {2 \over 3} \lambda_6 z_7\right) \partial_5
- \left( 3 z_3 + 5 \lambda_4 z_7 \right) \partial_7;
\\
\mathcal{L}_6 &= L_6
- \psi_5 \partial_1 - \psi_3 \partial_3 - \psi_1 \partial_5 + \\ & \quad
+ {5 \over 3} \lambda_8 z_3 \partial_1 + {16 \over 9} \lambda_8 z_5 \partial_3 - \left( \lambda_4 z_3 + 2 \lambda_6 z_5 -{8 \over 9} \lambda_8 z_7 \right) \partial_5 - \left(z_1 +3 \lambda_4 z_5 + 4 \lambda_6 z_7\right) \partial_7;
\\
\mathcal{L}_8 &= L_8
- \psi_7 \partial_1 - \psi_5 \partial_3 - \psi_3 \partial_5 - \psi_1 \partial_7 + \left({4 \over 3} \lambda_{10} z_3 - \lambda_{12} z_5\right) \partial_1 + \\ & \quad + {20 \over 9} \lambda_{10} z_5 \partial_3 - \left(\lambda_8 z_5 - {10 \over 9} \lambda_{10} z_7\right) \partial_5 - \left(\lambda_4 z_3 + 2 \lambda_6 z_5 + 3 \lambda_8 z_7\right) \partial_7;
\\
\mathcal{L}_{10} &= L_{10} 
- \psi_7 \partial_3 - \psi_5 \partial_5 - \psi_3 \partial_7  + \\ & \quad
+ (\lambda_{12} z_3 - 2 \lambda_{14} z_5 - \lambda_{16} z_7) \partial_1 + {5 \over 3} \lambda_{12} z_5 \partial_3 + {4 \over 3} \lambda_{12} z_7 \partial_5 - (\lambda_8 z_5 + 2 \lambda_{10} z_7) \partial_7;
\\
\mathcal{L}_{12} &= L_{12}
- \psi_7 \partial_5 - \psi_5 \partial_7 +\\ & \quad  + \left({2 \over 3} \lambda_{14} z_3 - 3 \lambda_{16} z_5 - 2 \lambda_{18} z_7\right) \partial_1 + \left({10 \over 9} \lambda_{14} z_5 - \lambda_{16} z_7 \right) \partial_3 + {14 \over 9} \lambda_{14} z_7 \partial_5 - \lambda_{12} z_7 \partial_7;
\\
\mathcal{L}_{14} &= L_{14}
- \psi_7 \partial_7 + \left({1 \over 3} \lambda_{16} z_3 - 4 \lambda_{18} z_5\right) \partial_1 + \left({5 \over 9} \lambda_{16} z_5 - 2 \lambda_{18} z_7\right) \partial_3 + {7 \over 9} \lambda_{16} z_7 \partial_5.
\end{align*}

Denote the Lie algebra with this generators by  $\mathscr{L}_\mathcal{L}$.

\begin{thm} \label{Tdif}
If the function $\varphi$ satisfies the system of heat equations in a nonholonomic frame for genus $4$, then the algebra $\mathscr{L}_\mathcal{L}$ is an algebra of derivations of the ring~$\mathcal{R}_\varphi$.
\end{thm}

\textit{The proof} follows the proof of Theorem 5.9 in \cite{BB19}.

\begin{cor} 
For $g = 4$ the algebra $\mathscr{L}_\mathcal{L}$ for $\varphi = \sigma$ is an algebra of derivations of $\mathcal{F}$.
\end{cor}

\vfill 
\eject

We set
\begin{align*}
w_{0,1} &= \psi_1, \qquad w_{0,3} = 3 \psi_3, \qquad w_{2k,5} = 5 \psi_5, \qquad w_{0,7} = 7 \psi_7,  
\\
w_{2,1} &= {1 \over 2} \psi_{111} + \psi_3 - {7 \over 9} \lambda_4 z_1, \\
w_{2,3} &= {1 \over 2} \psi_{113} - {4 \over 3} \lambda_4 \psi_1 + 3 \psi_5 + \left(3 \lambda_8 - {4 \over 3} \lambda_4^2\right) z_3,\\
w_{2,5} &= {1 \over 2} \psi_{115} - {8 \over 9} \lambda_4 \psi_3 + 5 \psi_7 + \left(5 \lambda_{12} - {8 \over 9} \lambda_4 \lambda_8\right) z_5, \\
w_{2,7} &= {1 \over 2} \psi_{117} - {4 \over 9} \lambda_4 \psi_5 + \left(7 \lambda_{16} - {4 \over 9} \lambda_4 \lambda_{12}\right) z_7.  
\\
w_{4,1} &= \psi_{113} + \psi_5 - {2 \over 3} \lambda_6 z_1 + \lambda_8 z_3,\\  
w_{4,3} &= \psi_{133} - 2 \lambda_6 \psi_1 + \lambda_4 \psi_3 + 3 \psi_7 + \lambda_8 z_1 +  \left(6 \lambda_{10} - 2 \lambda_4 \lambda_6\right) z_3 + 3 \lambda_{12} z_5,\\ 
w_{4,5} &= \psi_{135} - {4 \over 3} \lambda_6 \psi_3 + 3 \lambda_4 \psi_5 + 3 \lambda_{12} z_3 + \left(10 \lambda_{14} - {4 \over 3} \lambda_6 \lambda_8\right) z_5 + 5 \lambda_{16} z_7,\\
w_{4,5} &= \psi_{137} - {2 \over 3} \lambda_6 \psi_5 + 5 \lambda_4 \psi_7 + 5 \lambda_{16} z_5 + \left(14 \lambda_{18} - {2 \over 3} \lambda_6 \lambda_{12}\right) z_7.
\\
w_{6,1} &= {1 \over 2} \psi_{133} + \psi_{115} + \psi_7 - {5 \over 9} \lambda_8 z_1 + 2 \lambda_{10} z_3 + \lambda_{12} z_5;\\
w_{6,3} &= {1 \over 2} \psi_{333} + \psi_{135} - {5 \over 3} \lambda_8 \psi_1 + \lambda_4 \psi_5 + 2 \lambda_{10} z_1 - \left({5 \over 3} \lambda_4 \lambda_8 - 9 \lambda_{12}\right) z_3 + 6 \lambda_{14} z_5 + 3 \lambda_{16} z_7;\\
w_{6,5} &= {1 \over 2} \psi_{335} + \psi_{155} - {16 \over 9} \lambda_8 \psi_3 + 2 \lambda_6 \psi_5 + 3 \lambda_4 \psi_7 + \\
& \qquad + \lambda_{12} z_1 + 6 \lambda_{14} z_3 + \left(3 \lambda_4 \lambda_{12} - {16 \over 9} \lambda_8^2 + 15 \lambda_{16}\right) z_5 + 10 \lambda_{18} z_7;\\
w_{6,7} &= {1 \over 2} \psi_{337} + \psi_{157} -{8 \over 9} \lambda_8 \psi_5 + 4 \lambda_6 \psi_7 + 3 \lambda_{16} z_3 + 10 \lambda_{18} z_5 - \left({8 \over 9} \lambda_8 \lambda_{12} - 5 \lambda_4 \lambda_{16}\right) z_7.
\\
w_{8,1} &= \psi_{135} + \psi_{117} - {4 \over 9} \lambda_{10} z_1 + 3 \lambda_{12} z_3 + 2 \lambda_{14} z_5 + \lambda_{16} z_7;\\
w_{8,3} &= \psi_{335} + \psi_{137} - {4 \over 3} \lambda_{10} \psi_1 + \lambda_4 \psi_7 + \\
& \qquad + 3 \lambda_{12} z_1 - \left({4 \over 3} \lambda_4 \lambda_{10} - 12 \lambda_{14}\right) z_3 +\left(\lambda_4 \lambda_{12} + 9 \lambda_{16}\right) z_5 + 6 \lambda_{18} z_7;\\
w_{8,5} &= \psi_{355} + \psi_{157} + \lambda_{12} \psi_1 - {20 \over 9} \lambda_{10} \psi_3 + \lambda_8 \psi_5 + 2 \lambda_6 \psi_7 + 2 \lambda_{14} z_1 + \\
& \qquad +\left(\lambda_4 \lambda_{12} + 9 \lambda_{16}\right) z_3 - \left({20 \over 9} \lambda_8 \lambda_{10} - 2 \lambda_6 \lambda_{12} - 6 \lambda_4 \lambda_{14} - 20 \lambda_{18}\right) z_5 + 6 \lambda_4 \lambda_{16} z_7;\\
w_{8,7} &= \psi_{357} + \psi_{177} - {10 \over 9} \lambda_{10} \psi_5 + 3 \lambda_8 \psi_7 + \\
& \qquad + \lambda_{16} z_1 + 6 \lambda_{18} z_3 + 3 \lambda_4 \lambda_{16} z_5 - \left({10 \over 9} \lambda_{10} \lambda_{12} - 4 \lambda_6 \lambda_{16} - 10 \lambda_4 \lambda_{18}\right) z_7;
\end{align*}

\vfill
\eject

\begin{align*}
w_{10,1} &= {1 \over 2} \psi_{155} + \psi_{137}
- {1 \over 3} \lambda_{12} z_1 + 4 \lambda_{14} z_3 + 3 \lambda_{16} z_5 + 2 \lambda_{18} z_7;\\
w_{10,3} &=  {1 \over 2} \psi_{355} + \psi_{337}
- \lambda_{12} \psi_1 + \\
& \qquad + 4 \lambda_{14} z_1 - \left(\lambda_4 \lambda_{12} - 15 \lambda_{16} \right) z_3 + \left(2 \lambda_4 \lambda_{14} + 12 \lambda_{18}\right) z_5 + \lambda_4 \lambda_{16} z_7;\\
w_{10,5} &= {1 \over 2} \psi_{555} + \psi_{357}
+ 2 \lambda_{14} \psi_1 - {5 \over 3} \lambda_{12} \psi_3 + \lambda_8 \psi_7 + 3 \lambda_{16} z_1 + \\
& + \left(2 \lambda_4 \lambda_{14} + 12 \lambda_{18}\right) z_3 - \left({5 \over 3} \lambda_8 \lambda_{12} - 4 \lambda_6 \lambda_{14} - 9 \lambda_4 \lambda_{16}\right) z_5 + (2 \lambda_6 \lambda_{16} +6 \lambda_4 \lambda_{18}) z_7;\\
w_{10,7} &=  {1 \over 2} \psi_{557} + \psi_{377}
+ \lambda_{16} \psi_1 - {4 \over 3} \lambda_{12}  \psi_5 + 2 \lambda_{10} \psi_7 + \\
& \qquad + 2 \lambda_{18} z_1 + \lambda_4 \lambda_{16} z_3 + (2 \lambda_6 \lambda_{16} +6 \lambda_4 \lambda_{18}) z_5 - \left({4 \over 3} \lambda_{12}^2 - 3 \lambda_8 \lambda_{16} - 8 \lambda_6 \lambda_{18}\right) z_7;
\\
w_{12,1} &= \psi_{157} - {2 \over 9} \lambda_{14} z_1 + 5 \lambda_{16} z_3 + 4 \lambda_{18} z_5;\\
w_{12,3} &= \psi_{357} - {2 \over 3} \lambda_{14} \psi_1 + 5 \lambda_{16} z_1 -  \left({2 \over 3} \lambda_4 \lambda_{14} - 18 \lambda_{18}\right) z_3 + 3 \lambda_4 \lambda_{16} z_5 + 2 \lambda_4 \lambda_{18} z_7;\\
w_{12,5} &= \psi_{557} + 3 \lambda_{16} \psi_1 - {10 \over 9} \lambda_{14} \psi_3 + 4 \lambda_{18} z_1 + \\
& \qquad + 3 \lambda_4 \lambda_{16} z_3 - \left({10 \over 9} \lambda_8 \lambda_{14} - 6 \lambda_6 \lambda_{16} - 12 \lambda_4 \lambda_{18}\right) z_5 + \left(\lambda_8 \lambda_{16} + 4 \lambda_6 \lambda_{18}\right) z_7;\\
w_{12,7} &= \psi_{577} + 2 \lambda_{18} \psi_1 + \lambda_{16} \psi_3 - {14 \over 9} \lambda_{14} \psi_5 + \lambda_{12} \psi_7 +  \\
& \qquad + 2 \lambda_4 \lambda_{18} z_3 + \left(\lambda_8 \lambda_{16} + 4 \lambda_6 \lambda_{18}\right) z_7 - \left({14 \over 9} \lambda_{12} \lambda_{14} - 2 \lambda_{10} \lambda_{16} - 6 \lambda_8 \lambda_{18}\right) z_7;
\\
w_{14,1} &= {1 \over 2} \psi_{177} - {1 \over 9} \lambda_{16} z_1 + 6 \lambda_{18} z_3;\\
w_{14,3} &= {1 \over 2} \psi_{377} - {1 \over 3} \lambda_{16} \psi_1 + 6 \lambda_{18} z_1 - {1 \over 3} \lambda_4 \lambda_{16} z_3 + 4 \lambda_4 \lambda_{18} z_5;\\
w_{14,5} &= {1 \over 2} \psi_{577} + 4 \lambda_{18} \psi_1 - {5 \over 9} \lambda_{16}\psi_3 + 4 \lambda_4 \lambda_{18} z_3 - \left( {5 \over 9} \lambda_8 \lambda_{16} - 8 \lambda_6 \lambda_{18} \right) z_5 + 2 \lambda_8 \lambda_{18} z_7;\\
w_{14,7} &= {1 \over 2} \psi_{777} + 2 \lambda_{18} \psi_3 - {7 \over 9} \lambda_{16} \psi_5 + 2 \lambda_8 \lambda_{18} z_5 - \left({7 \over 9} \lambda_{12} \lambda_{16} - 4 \lambda_{10} \lambda_{18}\right) z_7.
\end{align*}

We obtain a genus $g=4$ analog of Theorems 4.1, 4.2 and 4.3 from \cite{BB19}.
\begin{thm}
For $g = 4$ a solution $\varphi$ of the system of heat equations \eqref{e3} gives a~solution $(\psi_1, \psi_3, \psi_5, \psi_7) = (\partial_1 \ln \varphi, \partial_3 \ln \varphi, \partial_5 \ln \varphi, \partial_7 \ln \varphi)$ of the system of nonlinear differential equations that we call an analog of the Burgers equation for $g = 4$:
\begin{equation} \label{Bu2}
\mathcal{L}_{2k} (\psi_1, \psi_3, \psi_5, \psi_7) = (w_{2k,1}, w_{2k,3}, w_{2k,5}, w_{2k,7}), \qquad k = 0,1,2,3,4,5,6,7.
\end{equation}
\end{thm}
\textsc{The proof} is the result of direct computation.

\vfill

\text{ }

\end{document}